\documentclass[aps,twocolumn,secnumarabic,balancelastpage,amsmath,amssymb,nofootinbib,floatfix,prd,longbibliography]{revtex4-1}
\usepackage{graphicx}      
\usepackage{physics}
\usepackage{verbatim}
\usepackage{colonequals}
\usepackage{siunitx}
\usepackage[T1]{fontenc}
\usepackage{currvita}
\usepackage[table]{xcolor}
\usepackage{booktabs}
\usepackage{adjustbox}
\usepackage{multirow}
\usepackage{bm}   
\usepackage{placeins}
\usepackage{dblfloatfix}
\usepackage[colorlinks=true]{hyperref} 
\usepackage[bottom]{footmisc}
\usepackage{subcaption}
\usepackage{ragged2e}
\DeclareCaptionJustification{justified}{\justifying}
\usepackage{amsmath, amssymb, amsthm}
\newtheorem{thm}{Theorem}

\usepackage{mathtools}
\usepackage{tcolorbox}

\date{}

\theoremstyle{definition}

\newcommand{\mrm}{\mathrm}
\newcommand{\ovrl}{\overline}
\newcommand{\mc}{\mathcal}
\newcommand{\eqb}{\begin{equation}}
\newcommand{\eqe}{\end{equation}}

\begin{document}
\title{Shaping Long-lived Electron Wavepackets for Customizable Optical Spectra}
\author{Rumen Dangovski}
\email{rumenrd@mit.edu}
\affiliation{Department of Physics, Massachusetts Institute of Technology, Cambridge, Massachusetts 02139, USA}
\author{Nicholas Rivera}
\affiliation{Department of Physics, Massachusetts Institute of Technology, Cambridge, Massachusetts 02139, USA}
\author{Ido Kaminer}
\affiliation{Department of Electrical Engineering, Technion -- Israel Institute of Technology, Haifa 32000, Israel}
\author{Marin Solja\v{c}i\'{c}}
\affiliation{Department of Physics, Massachusetts Institute of Technology, Cambridge, Massachusetts 02139, USA}
\date{\today}

\begin{abstract}
Electrons in atoms and molecules are versatile physical systems covering a vast range of light-matter interactions, enabling the physics of Rydberg states, photon-photon bound states, simulation of condensed matter Hamiltonians, and quantum sources of light. A limitation on the versatility of such electronic systems for optical interactions would appear to arise from the discrete nature of the electronic transitions and from the limited ionization energy, constraining the energy scale through which the rich physics of bound electrons can be accessed. In this work, we propose the concept of shaping spatially confined electron wavepackets as superpositions of extended states in the ionization continuum. These wavepackets enable customizable optical emission spectra transitions in the eV-keV range. We find that the specific shaping lengthens the diffraction lifetime of the wavepackets in exchange for increasing their spatial spreads. Finally, we study the spontaneous radiative transitions of these excited states, examining the influence of the wavepacket shape and size on the radiation rate of the excited state. We observe that the rate of radiative capture is primarily limited by the diffraction lifetime of the wavepacket. The approach proposed in this work could have applications towards developing designer emitters at tunable energy and length scales, potentially bringing the physics of Rydberg states to new energy scales and extending the range and versatility of emitters that can be developed using shaped electrons.
\end{abstract}

\maketitle
The rich physics of bound electrons in atoms and molecules is typically limited by the discrete nature of the energy spectrum and by the ionization energy threshold. This is why many phenomena considered in atomic physics are usually confined to the IR-UV spectral range. The upper limit is set by the typical ionization energies of the outer valence electrons. Going beyond the ionization energy usually involves electron states that are not bound in space---extended states---which have positive energy and are therefore not considered relevant for atomic physics phenomena. To deal with this problem, one could engineer the potential to allow for the existence of bound states in the continuum (BICs) \cite{hsu2016}. Von Neumann and Wigner were the first to introduce such a construction, defying the conventional wisdom that bound states must be spectrally separated from the extended states \cite{neumann1929}. However, the kind of potentials which support BICs have to be specially designed, while the potentials of atoms and molecules generally cannot be designed \cite{plotnik2011, hsu2013, zhen2014, rivera2016}.

There is another approach by which electron states of positive energy can behave as bound states, which is by shaping the electron wavepacket as a superposition of extended states, thus creating a localized state. For example, in 1979 Berry and Bal\'{a}zs proposed a solution of the free-space Schr\"{o}dinger equation that appears to be localized in 1D and whose shape remains time-invariant as with bound electron states \cite{berry1979}. In recent years, this idea has been extensively explored in the optics community, with paraxial and non-paraxial optical beams \cite{siviloglou2007-2, siviloglou2007, kaminer2012}. 

More generally, these kinds of optical beams, such as 
Bessel (\cite{durnin1987i, durnin1987}) and Airy (\cite{siviloglou2007, siviloglou2007-2}) beams, are propagation-invariant wave functions that also accelerate in the absence of external force, and have the intriguing property of self-healing--restoring their original shape after encountering an obstacle \cite{broky2008}. While the above wavepackets are mostly only localized in 2D, the concept of optical beam shaping can also be applied to create ``light-bullets'', which are localized in 3D \cite{chong2010, liang2017}. Importantly, due to the mathematical analogies between optical wavepackets and electron wavepackets, similar concepts are applicable to electrons \cite{bliokh2011, noa2013, grillo2014, harris2015}. 

However, the well-known problem with all of the above time-invariant or propagation-invariant wavepackets is that their probability density is not square-integrable (hence without a physical interpretation as a probability density). To circumvent this, one can truncate the wavepackets at a finite distance. In this case, there is no exact time-invariance or propagation-invariance, but the non-diffracting properties are still present for a finite time and distance, which can be long enough for the desired interaction in experiments \cite{siviloglou2007, mathis2012, noa2013, schley2014}. Another approach that makes the wavepacket square-integrable is by shaping a superposition of a range of energies/frequencies, which can create a localized wavepacket, in exchange for limiting its range/duration for which it is non-diffracting \cite{liang2017}. Therefore, wavepacket shaping in space and time offers localized and long-lived electron wavepackets in the continuum of energy levels, created from superpositions of extended states. 

With the above wavepackets in mind, we now ask: can shaped electron wavepackets be used to mimic optical phenomena, usually accessible only in bound electron systems, such as electron transitions by light emission/absorption? Can shaped electrons access light-matter interactions beyond the ionization threshold that limits bound electron systems? For example, can one create engineered spontaneous emission dynamics (engineered rates, engineered optical spectra, etc.) by shaping a positive energy wavepacket of tailored shape and energy?

\begin{figure}[!t]
    \includegraphics[width=8.6cm]{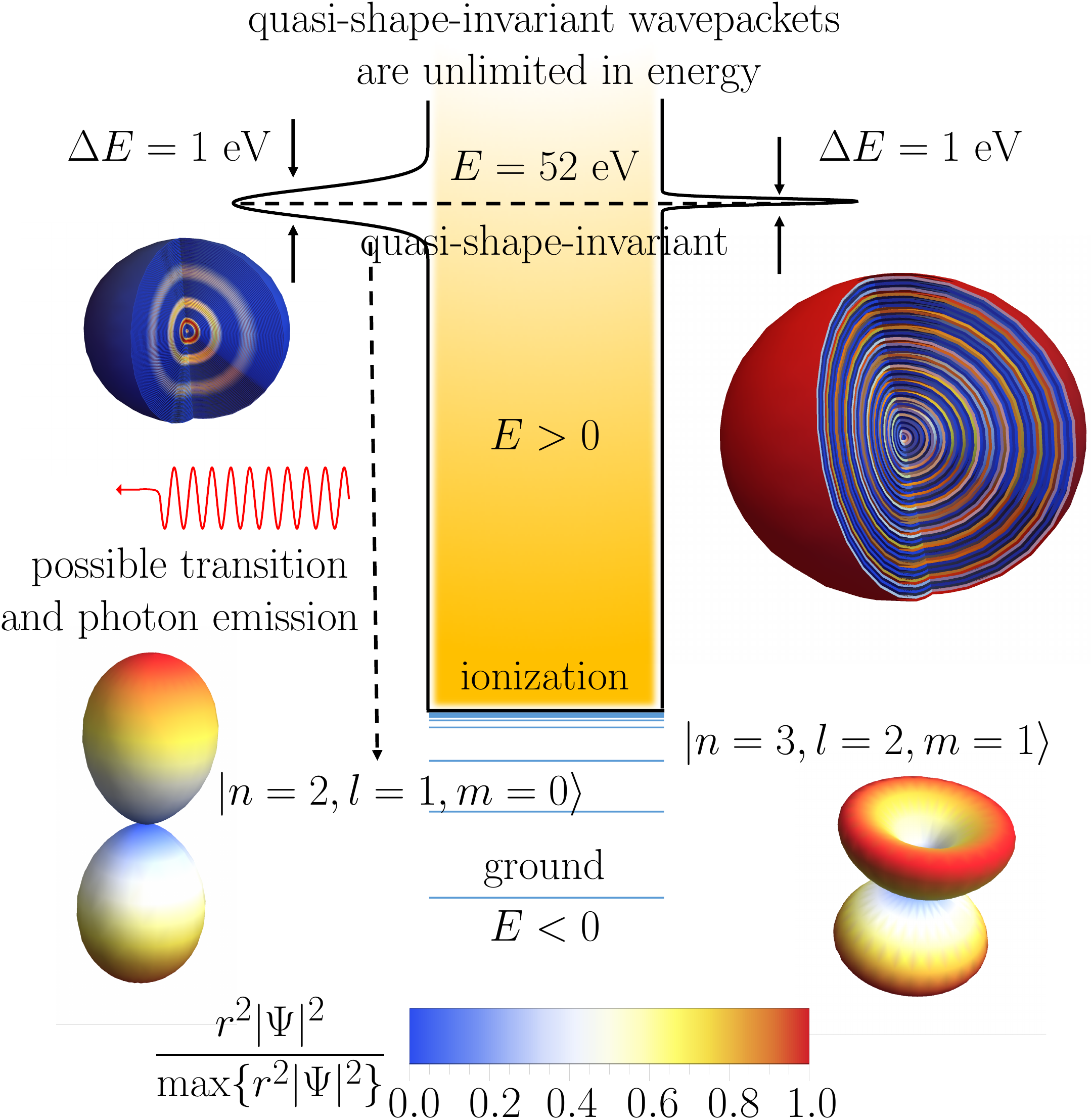}
    \caption{
        \textbf{Shaping of electron states in the continuum of energy levels of the hydrogen atom creates localized and quasi-shape-invariant high-energy wavepackets.} As opposed to the bound states, the electron states in the continuum of the energy spectrum are described by Whittaker functions, thus we call our wavepackets ``Whittaker wavepackets.'' The possibility of decay to bound states enables a photon emission with energy higher than the ionization threshold. We present truncated temperature maps of the typical shapes of the probability densities for both Whittaker wavepackets and bound states. As we decrease the energy spread $\Delta E$, the wavepacket spreads out in space, hence the electron is more likely to be found farther from the origin. 
        \label{main:fig:constructions}
    }
\end{figure}

Here, we propose and study the shaping of shape-invariant wavepackets that simultaneously suppress their own diffraction, while enabling access to a customizable spectrum of transitions, ranging from the visible to the hard-X-ray, via radiative decay to bound states. We develop the analytic tools to maintain the shape-invariance of shaped electron states and the analytic tools to calculate the radiative transitions of such states into bound states. Our methods can be extended to a variety of potentials, including the transitions of free electrons illuminated by general time-dependent fields. Specifically for attractive static potential, we derive the non-diffracting electron wavepackets and study their dynamics. For example, we find that a non-diffracting shape of the wavepacket also affects the behavior of the electron in an attractive potential. We show that the presence of a Coulomb potential changes the physics of the system drastically by allowing propagating wavepackets to decay to bound states through radiative capture. We monitor the ``competition'' between spontaneous emission and diffraction by developing a Fermi Golden Rule (FGR) formalism which quantifies the rate of decay by the excited wavepacket. We find that in all these cases the electron state lifetime is limited by the diffraction dynamics of its wavepacket. Even though, in general, the electron wavepackets we consider are diffracting and subject to spontaneous emission, we find parameters that can suppress those effects significantly. Hence, we concisely refer to the constructed states, which are almost propagation-invariant and time-invariant, as being \textit{quasi-shape-invariant}.

To illustrate the concept of quasi-shape-invariant quantum wavepackets above the ionization threshold, we begin from 
the textbook example of the hydrogen atom, consisting of the Schr\"{o}dinger equation with the Coulomb potential $V(r) = -e^2/4\pi\epsilon_0 r$. The hydrogen atom is one of the most famous problems in quantum mechanics; its electron bound states are well-studied, and analytic expressions for the extended states exist in the literature \cite{landau1981}. We introduce a dimensionless parameter $x$, defined by $x=(2/a_0)r$ with spherical radius $r$ and the Bohr radius $a_0$. Likewise, let $\kappa=ka_0$ be the dimensionless parameter from momentum $k$. 
Figure \ref{main:fig:constructions} presents the shaping of the electron wavepacket that is created from superpositions of eigenstates at positive energies, which are called the Whittaker functions $w_\kappa(x,t)$ and can be found in \cite{abramowitz1972, zwillinger1997}. These eigenstates are of the form \eqb \label{main:eq:solution_schroedinger} w_\kappa(x,0) = \frac{4i\kappa^2 e^{-i\kappa x} }{\pi\csch\left( \pi / 2\kappa \right)}\int_0^1 e^{2i\kappa x s} {\left(\frac{s}{1-s}\right)}^{\frac{i}{2\kappa}} \dd s, \eqe and $w_\kappa(x,t)=w_\kappa(x,0)\exp(-i\omega t \kappa^2),$ where the time evolution frequency is given by $\omega\kappa^2$ with $\omega = 2e^2/a_0\hbar \approx 82 \ \si{\femto \second^{-1}}.$ 

For simplicity, we focus on a spherically symmetric wavepacket with a Gaussian weighting over momentum space $\Psi_{E,\Delta E}(r,t)$, given as follows \eqb \label{main:eq:whittaker_packet} \Psi_{E,\Delta E}(r,t)=\mc{N} \int_{-\infty}^{\infty}  e^{-( \kappa-\mu)^2/2\sigma^2}w_\kappa(x,t)\dd \kappa, \eqe where $\mc{N}$ is a normalization constant, $\mu$ and $\sigma$ are the mean and spread (standard deviation) of momentum. In SI units the energy $E$ is parameterized as  $E(\kappa) = (2e^2/4\pi\epsilon_0 a_0)\kappa x^2$ via the dimensionless $\kappa$. We denote $E=E(\mu)$ and $\Delta E$ for the spread (standard deviation). We call the resulting wavepacket \eqref{main:eq:whittaker_packet}, which is a superposition of the eigenstates \eqref{main:eq:solution_schroedinger}, a ``Whittaker wavepacket.''

\begin{figure}[!t]
    \includegraphics[width=8.6cm]{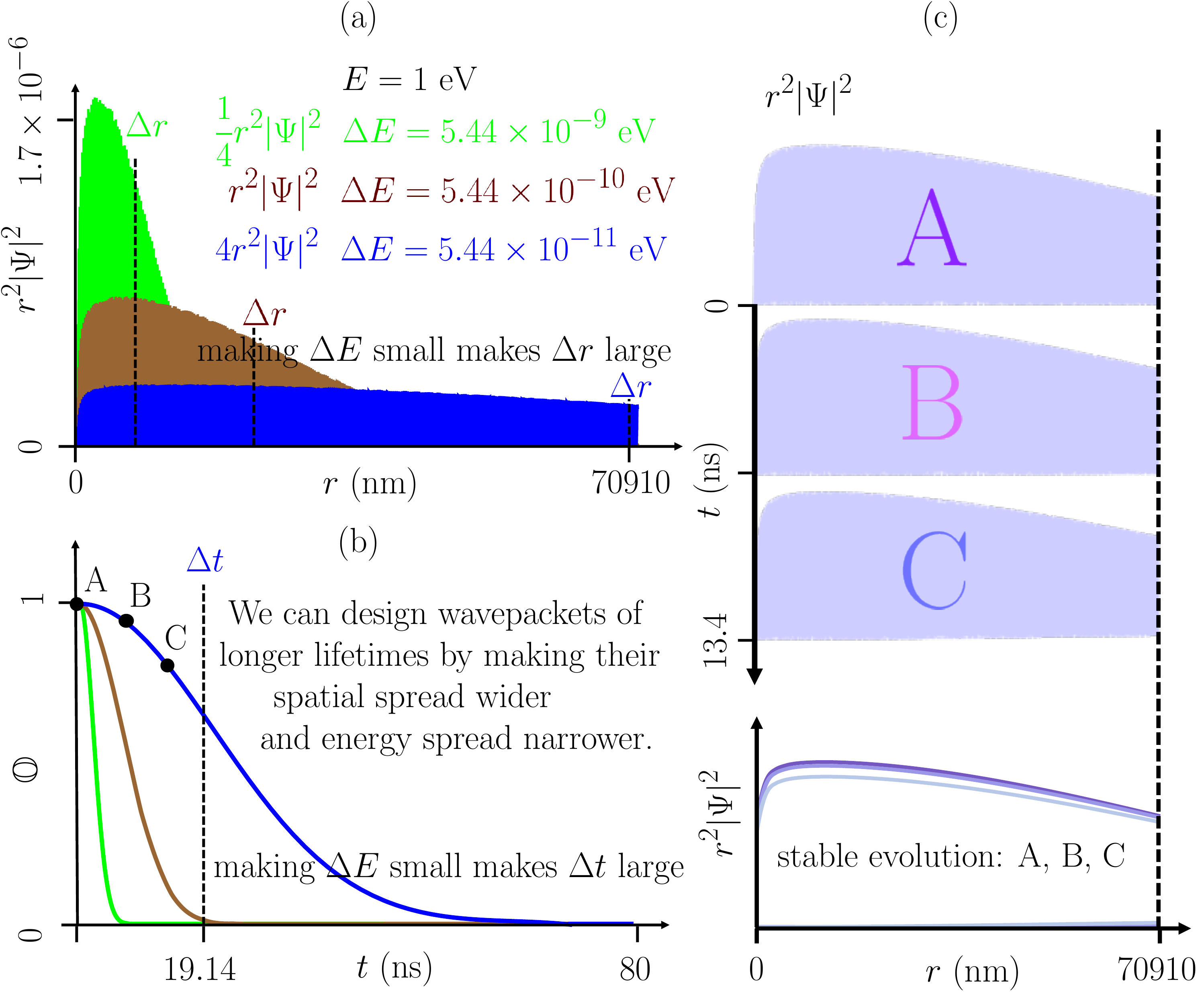}
    \caption{
        \textbf{By shaping the (Whittaker) electron wavepacket we can prolong the lifetime to a desired duration in exchange for increasing the spatial spread.} (a) As we decrease the energy spread $\Delta E$, the probability density spreads out farther in space according to \eqref{main:eq:spread}. Thus there is less probability to find the electron near the origin. (b) The benefit of making $\Delta E$ small, however, is that the lifetime grows according to \eqref{main:eq:lifetime}. Here $\Delta t = 19.14 \ \si{\nano \second}$ and the colors match the Whittaker wavepackets in (a). (c) The envelope of the wavepacket at three points in time A,B,C (also marked in (b)), showing shape-invariant dynamics and slow diffraction over tens of \si{\nano \second}. We view the large $\Delta t$ as the quasi-shape-invariant evolution of the long-lived Whittaker electrons. 
        \label{main:fig:stability}
    }
\end{figure}
Having shaped the wavepackets, we now characterize their spatial and temporal dynamics. One way to do this is to define a spatial spread $\Delta r$, given by a standard deviation  $\Delta r = \sqrt{\mathrm{var}(\mathrm{envelope}(\Psi_{E,\Delta E}(r,0)))}$, and a diffraction lifetime $\Delta t$, also defined as a standard deviation $\Delta t = \sqrt{\mathrm{var}(\mathbb{O}(t))}$ of an overlap function
$\mathbb{O}(t) = |\int_0^{\infty} \Psi^*(r,0)\Psi(r,t)r^2 \dd r|^2.$ Figure \ref{main:fig:stability} (a) reveals that as we decrease the energy spread $\Delta E$, then $\Delta r$ increases. Manifestly, this should hold true as a narrower $\Delta E$ yields a wavepacket approaching an extended state. The analytic form of the the Whittaker wavepacket \eqref{main:eq:whittaker_packet} suggests a functional form, for which we numerically obtain the spatial spread formula 
\eqb \label{main:eq:spread}
\Delta r \approx \frac{2.471 \ a_0}{\sqrt{(\Delta E) / \si{\electronvolt}}}.
\eqe
Furthermore, the overlap $\mathbb{O}(t)$ resembles a Gaussian function as shown in figure \ref{main:fig:stability} (b). In a similar fashion, as we decrease $\Delta E$, the diffraction lifetime increases, because we approach an extended state, which is also stationary. A stationary phase argument yields the diffraction lifetime formula 
\eqb \label{main:eq:lifetime}
\Delta t \approx  \frac{0.136 \ \si{\electronvolt} \cdot \si{\femto \second}}{\sqrt{(\Delta E) E}}.
\eqe
Further discussion on the ansatzes of formulas \eqref{main:eq:spread} and \eqref{main:eq:lifetime} is deferred to the SM sections II and in SM section IV we comment on our numerical fits. In theory we can make $\Delta t$ very large if we have $\Delta E$ as small as possible. Naturally, as this limit is taken, $\Delta t$ increases without bound, which is illustrated in figure \ref{main:fig:stability} (c). In the figure we observe the quasi-shape-invariant nature of the Whittaker wavepackets, that is the flexibility to tune large $\Delta t$ by customizing the parameters $E$ and $\Delta E$.
\begin{figure}[!t]
    \includegraphics[width=8.6cm]{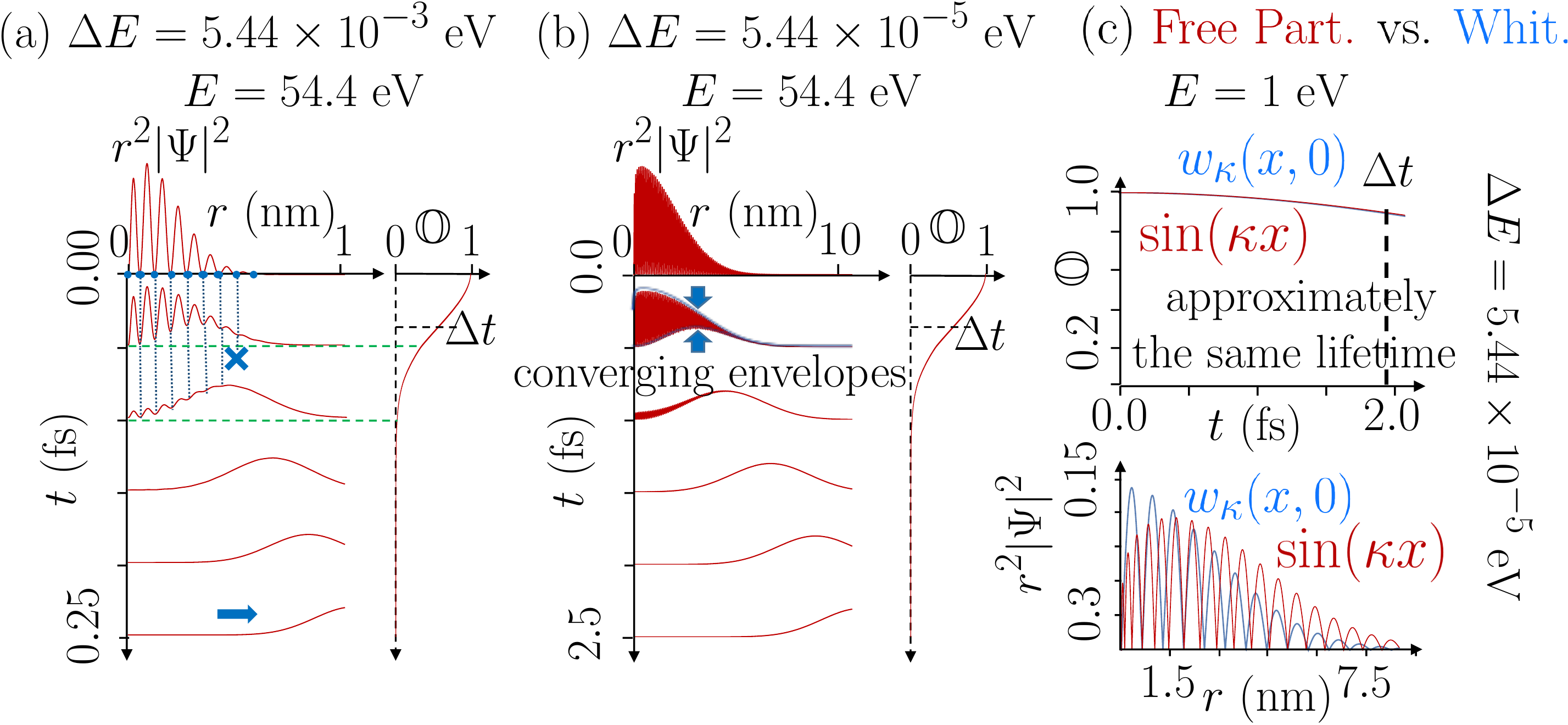}
    \caption{
        \textbf{The dynamics of the Whittaker wavepacket.} (a) The wavepacket begins at $t=0$, with a set of nodes marked by blue dots. The amplitude of the spatial oscillations decreases in time until the nodes vanish and the wavepacket starts to spread out in space. (b) A larger range in space that includes more oscillations (only the envelope is plotted) and has a longer lifetime. The upper and lower wavepacket envelopes converge as time evolves. Notice the similarity between the dynamics of the envelopes in (a) and (b). (c) The geometry of the Whittaker wavepacket is roughly the same as that of the free particle (with Bessel functions as modes, so at large spatial distances they behave approximately as sine funcitons; here $l=0$). This similarity shows that our methods and conclusions also apply for free electron wavepackets (possibly with a time dependent perturbation like a pulsed laser excitation).
        \label{main:fig:dynamics-of-stability}
    }
\end{figure}

We now study the profile of the Whittaker wavepacket and how it evolves in time in a shape-invariant manner. Just as wavepacket shaping is known to extend the lifetime of decaying particles in the Dirac equation \cite{kaminer2015}, our shaped electron can exhibit a longer diffraction lifetime, and as we show below, also a longer radiative lifetime. For example, figure \ref{main:fig:dynamics-of-stability} (a) presents the Whittaker wavepacket's profile at time zero, marking its nodal structure that highlights the shape-invariant properties, in a similar fashion to the shaped packets in related works \cite{berry1979, durnin1987i, durnin1987, siviloglou2007, siviloglou2007-2, broky2008, kaminer2012, kaminer2015}. As time evolves, the nodes vanish sequentially, which happens via ``lifting'' of the wavepacket profile due to continuity. In this way the diffraction eventually destroys the original nodal structure, and having no nodes means that the wavepacket is free to propagate in space like a free wavepacket, as shown in figure \ref{main:fig:dynamics-of-stability} (b). In figure \ref{main:fig:dynamics-of-stability} (b) the upper and lower envelopes converge, which destroys the nodal structure. This process of ``node-lifting'' bottlenecks the diffraction of the wavepacket and enables the quasi-shape-invariant property. A proof of the existence of nodes is given in SM section II, theorem 3. Finally, in figure \ref{main:fig:dynamics-of-stability} (c) we show that free particle wavepackets also have a similar nodal structure at time zero, which yields approximately the same overlap function as that of Whittaker wavepackets. A way to see why this holds true is to look at the 
limiting behavior of the extended mode \eqref{main:eq:solution_schroedinger} for large $r$ ($x \gg 1$)  \cite{zwillinger1997}: $w_\kappa(x,0) \sim \exp(i\kappa x)\exp(i(1/2\kappa)\ln(x))/x.$ The log dependence on $x$ is due to the Coulomb potential. Thus the electron wavepacket does not approximate a free particle wave exactly; however, at large values of $x$ the oscillations from $\exp(-i\kappa x)$ are dominant and form the free particle modes. This means that the Whittaker wavepacket is similar to that of a free particle; hence their time-evolutions should be similar.

\begin{figure}[!t]
    \includegraphics[width=8.6cm]{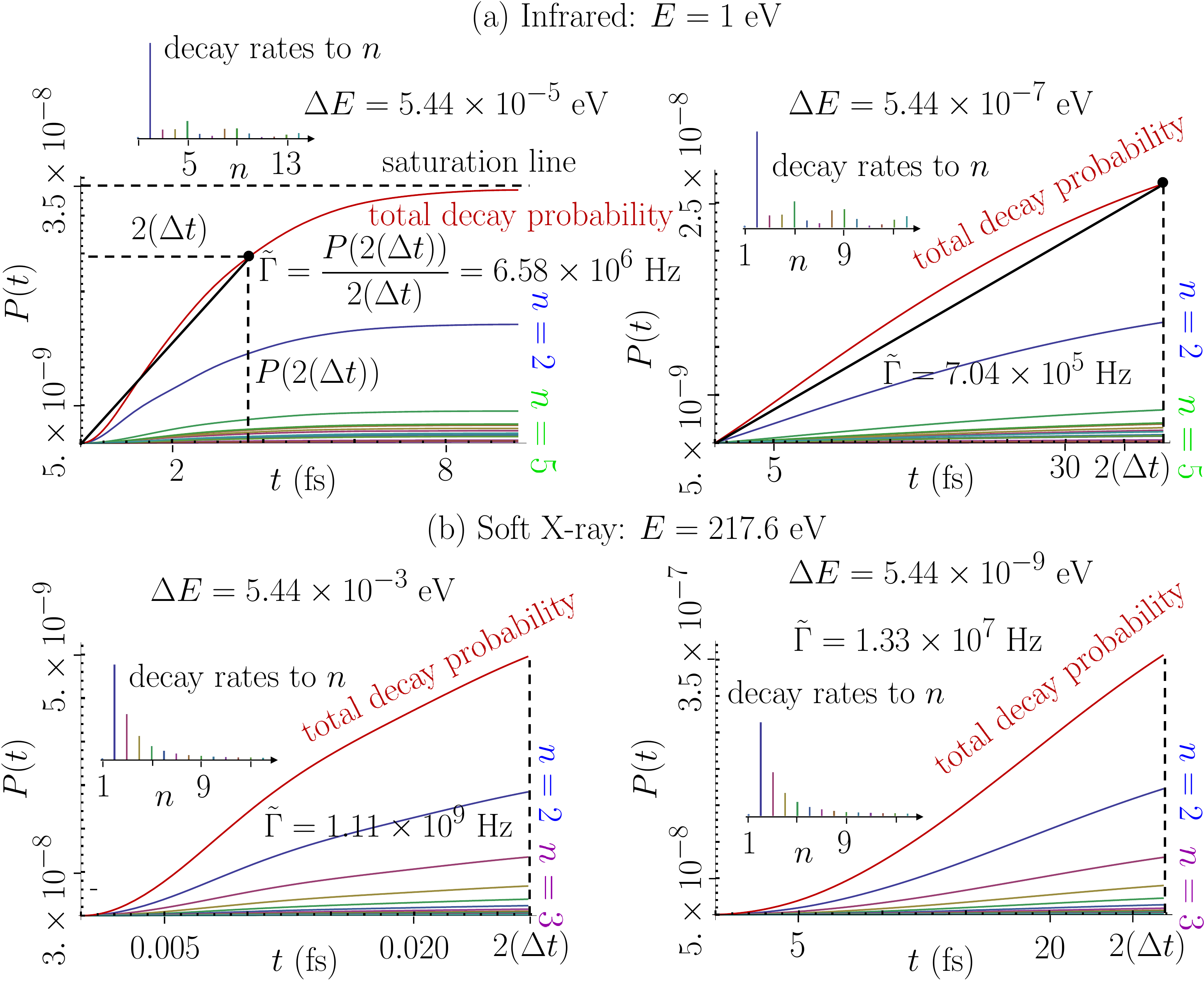}
    \caption{
        \textbf{The radiative decay of the Whittaker wavepackets to bound states.} (a) The decay probability reaches a saturation line. As time evolves, the electron is farther away from the origin, thus the effect of the Coulomb potential is reduced, and the instantaneous rate of decay converges to zero. The profile of decay rates to individual bound states follows an oscillatory pattern, with the highest decay rate to $n=2$. (b) The profiles of the decay rates are monotonously decreasing in comparison to the profiles in (a). Note that (b) may have quantitative corrections due to beyond-dipole corrections arising from the short wavelength of the electromagnetic field since for our FGR calculations we use the dipole approximation. The dynamics and lifetime of the electron wavepacket are dominated by the diffraction and not by the radiative decay, as the diffraction lifetime is generally shorter than the radiative (FGR) lifetime.
        \label{main:fig:spontaneous-emission}
    }
\end{figure}

Importantly, however, the physics of free particles is different from that of hydrogenic electrons. The potential may cause scattering processes between the proton and electron, which in turn could lower the energy of the superposition through radiative decay, and thus reduce its stability. We now quantitatively evaluate the competition between diffraction and radiative decay by quantifying the radiative decay. For that purpose, to get the probability of decay $P_n(t)$ to a bound state $\ket{n}$ at a given time $t$, we develop a FGR formalism (described in the Methods section)  through the $S$-matrix approach \cite{landau1981}. The total decay probability is simply the sum $P(t)=\sum_{n=1}^{\infty}P_n(t)$. Having the total probability, we define a figure of merit \textit{average rate} of decay over the time of two diffraction lifetimes $\Delta t$ via the formula $\tilde{\Gamma}=P\left(2(\Delta t)\right)/2(\Delta t)$, and analyze it in figure \ref{main:fig:spontaneous-emission}. In (a) we observe that the probability saturates within a time of the order of $\Delta t$. Moreover, in both cases, (a) and (b), we find that the diffraction is the dominant factor for the lifetime of the Whittaker wavepackets. In particular, if we hypothesize a lifetime $1/\tilde{\Gamma}$ due to the spontaneous emission then we see that in general $\Delta t \ll 1/\tilde{\Gamma}.$ For example, if $E = 1 \ \si{\electronvolt}$ and $\Delta E = 5.44 \times 10^{-5} \ \si{\electronvolt}$, then $\tilde{\Gamma} = 6.58 \times 10^{6} \ \si{\hertz}$. This yields a hypothetical lifetime of $1 / 6.58 \times 10^{6} \ \si{\hertz} = 152 \ \si{\nano \second}$. As a comparison, for the same case, $\Delta t = 1.91 \ \si{\femto \second}.$ We also observe the pattern that $\tilde{\Gamma}$ grows with an increase in $\Delta E.$ Many other such examples of low-magnitude $\tilde{\Gamma}$ are given in SM section IV. 

Since the radiative lifetime is much longer than the diffraction lifetime, we conclude that formulas \eqref{main:eq:spread} and \eqref{main:eq:lifetime} give a good parametrization of the stability and the large $\Delta t$ vs. large $\Delta r$ trade-off of the Whittaker wavepackets, and show how those properties can be customized.

\begin{table}[!b]
\scriptsize
\begin{tabular}{lr|c|c|c} 
             & & visible light & soft X-ray & hard X-ray \\
fixed & determined & $1$ \si{\electronvolt} & $200$ \si{\electronvolt} &
$10000$ \si{\electronvolt} \\\midrule
\multirow{2}{*}{$\Delta t = 53 \ \si{\atto \second}$} &$\Delta E$ (\si{\electronvolt}) & 6.6 & $\mathbf{3.3 \times 10^{-2}}$ & $\mathbf{6.6 \times 10^{-4}}$ \\ 
& $\Delta r$ (\si{\nano \meter}) & $5.1 \times 10^{-2}$ & $\mathbf{7.2 \times 10^{-1}}$ & \textbf{5.1} \\
 \midrule
 \multirow{2}{*}{$\Delta t = 100 \ \si{\atto \second}$} &$\Delta E$ (\si{\electronvolt}) & 1.9 & $\mathbf{9.3 \times 10^{-3}}$ & $\mathbf{1.9 \times 10^{-4}}$ \\ 
& $\Delta r$ (\si{\nano \meter}) & $\mathbf{9.6 \times 10^{-2}}$ & \textbf{1.4} & \textbf{9.6} \\
 \midrule 
 \midrule
 \multirow{2}{*}{$\Delta r = 143 \ \si{\nano \meter}$} &$\Delta E$ (\si{\electronvolt}) & \multicolumn{3}{c@{\quad}}{$\mathbf{8.4 \times 10^{-7}}$}      \\ 
& $\Delta t$ (\si{\atto \second}) & $\mathbf{1.5 \times 10^5}$ & $\mathbf{1.1 \times 10^{
4}}$ & $\mathbf{1.5 \times 10^3}$ \\
 \midrule
 \multirow{2}{*}{$\Delta r = 10 \ \si{\nano \meter}$} &$\Delta E$ (\si{\electronvolt}) & \multicolumn{3}{c@{\quad}}{$\mathbf{1.7 \times 10^{-4}}$}      \\ 
& $\Delta t$ (\si{\atto \second}) & $\mathbf{1.0 \times 10^{4}}$ & $\mathbf{7.4 \times 10^{2}}$ & $\mathbf{1.0 \times 10^{2}}$ \\\bottomrule
\end{tabular}
\caption{\textbf{Numerical experiments for the spread-lifetime trade-off of Whittaker wavepackets.} 
For a given energy $E$ (visible light, soft X-ray or hard X-ray) a parameter is fixed (either $\Delta t$ or $\Delta r$). In bold are all parameter cases that fit into attophysics and Rydberg atom scales. In the last two cases, there is a single number for $\Delta E$ since $\Delta r$ is independent of $E$.
\label{main:table:attophysics-rydberg-atoms}
}
\end{table}

We now present quantitative examples  of parameters achievable with Whittaker wavepackets, summarized in table \ref{main:table:attophysics-rydberg-atoms}. For a quasi-shape-invariant state to be considered stable enough for optical transitions, its diffraction lifetime $\Delta t$ should be longer than the duration of the optical cycle of the photon emitted from the transition. For states designed to have transitions in the X-ray frequencies, the shortest lifetime we consider in Table I is 53 as (comparable to the shortest X-ray pulse duration measured \cite{li2017} with high harmonic generation \cite{shan2001, ishii2014, keathley2016}). For the largest $\Delta r$, considered in table \ref{main:table:attophysics-rydberg-atoms}, we take 143 \si{\nano \meter} (comparable to the recently observed \cite{palmer2017} Rydberg state at the $n=52$ state). We note that by shaping electrons in transmission electron microscopes, coherent wavepackets over spatial extents of tens of microns have been observed \cite{shiloh2014}, hence one can consider much wider electron wavepackets as well. 
To give specific examples from table \ref{main:table:attophysics-rydberg-atoms}, if we fix $\Delta t = 53 \ \si{\atto \second}$, then according to formulas \eqref{main:eq:spread} and \eqref{main:eq:lifetime} a soft X-ray Whittaker wavepacket ($E = 200 \ \si{\electronvolt}$) has an energy spread $\Delta E = 0.033 \ \si{\electronvolt}$ and a spatial spread $\Delta r = 0.72 \ \si{\nano \meter}$, which all fit within the limits established here. If we consider a hard X-ray Whittaker wavepacket ($E=10 \ \si{\kilo \electronvolt}$) then the energy spread has to shrink to $\Delta E = 66 \ \si{\micro \electronvolt}$; for the same lifetime, however, $\Delta r$ gets to 5.1 \si{\nano \meter}, much smaller than 143 \si{\nano \meter}. We also studied Whittaker wavepackets with lifetimes on the order of nanoseconds, which require spread $\Delta r$ on the order of microns. These widely different regimes of parameters show the strength of our approach coming from the simple trade-off between $\Delta t$ and $\Delta r$.

In summary, we have proposed shaping quasi-shape-invariant wavepackets in the continuum of the hydrogem atom, the concept of which may lift the energy limitations of accessible light-matter interactions above the  ionization levels for atoms and molecules. By characterizing the dynamics of the wavepackets, we describe the large lifetime vs. large spatial spread trade-off, intrinsic to the quasi-shape-invariant electron wavepackets. These wavepackets exhibit unique phenomena in their decay dynamics, such as a stark change from radiative decay similar to bound states at short times, to saturation at long times (once the electron's probability density spreads away). It will be interesting to study how more complicated shapes (e.g. Airy) of the wavepackets alter the spontaneous radiative transition rates. Our methods, and specifically the time-dependent FGR formalism, can be applied to a variety of other systems. For instance, a good candidate is the system of shaped free electrons that interact with time-dependent potentials \cite{barwick2009, armin2015, echternkamp2016, ryabov2016}. The shaping of quasi-shape-invariant states can bring atomic physics phenomena to new energy ranges such as soft and hard X-rays. Utilizing these phenomena might introduce new quantum light sources and other applications to a diversity of physical systems: including various Dirac-particles and free electrons under strong fields, as well as other wave systems that often describe analogous physics, such as water waves, acoustic waves on membranes, and electromagnetic waves.

We would like to thank Pamela Siska for useful comments on the paper and Thomas Beck for useful discussions. I. K. was supported by the Marie Curie Grant 328853-MC-BSiCS. N. R. was supported by the Department of Energy Fellowship DE-FG02-97ER25308. Research was supported as part of the Army Research Office through the Institute for Soldier Nanotechnologies under contract no. W911NF-13-D-0001 (photon management for developing
nuclear-TPV and fuel-TPV mm-scale-systems). Also supported as part of the S3TEC, an Energy Frontier Research Center funded by the US Department of Energy under grant no. DE-SC0001299 (for fundamental photon
transport related to solar TPVs and solar-TEs).

\section*{Methods}
The nature of our analysis (both analytic and numerical) is universal, and can be generalized to calculate the transitions of superposition states under any potential, including time-dependent potentials. Specifically, the derivation of the Whittaker modes, the shaping of the wavepackets, the analytic monitoring of the evolution, and the shaping that enables the quasi-shape-invariant properties of the superpositions (see further details in SM sections I 
and II) can also be reproduced for electrons in the vicinity of other potentials. However, most cases would require finding the extended states numerically or working with arbitrary wavepackets, without finding the eigenstates at all. Furthermore, to obtain an analytic expression for the probability of transition decay from wavepackets in the continuum to bound states, we use FGR calculations that, in theory, can be applied for any interaction potential. 

Our FGR calculations are based on a QED formalism and the $S$-matrix approach \cite{landau1981}, for which the infinitesimal probability of transition from the initial state $\ket{i}$ to the final state $\ket{f}$, through the emission of a photon with momentum $\vb{k}$ and polarization $\lambda$, is given by $\dd P_{\mrm{fi}}(\vb{k},\lambda) = (V\dd^3 \vb{k})/(2\pi)^3|S_{\mrm{fi}}(\vb{k},\lambda)|^2$ for a finite volume $V;$ here $S_{\mrm{fi}}=\bra{\mrm{f}} T e^{-\frac{i}{\hbar} \int_0^t \hat{H}_{\text{int}} \dd t'} \ket{\mrm{i}}$ is the matrix element of the time-ordered unitary evolution operator of the EM interaction Hamiltonian $\hat{H}_{\text{int}}[\psi] = -(i\hbar e)/m_e \int \dd^3 \vb{x} \psi^* \vu{A}(\vb{x},t) \cdot \grad \psi$, which for the wavepacket $\psi$ is defined via the vector potential $\vu{A}$ and the mass of the electron $m_e$. Up to first order, we derive the universal expression $S_\mrm{fi}(\vb{k},\lambda)=-e/m_e\sqrt{\hbar/2\epsilon_0 \omega_{k} V} \int_0^t \dd t' e^{i\omega_{k}t'} \vu{\varepsilon}_{\vb{k}\lambda} \cdot
\int \dd^3 \vb{x} \bra{f}(t')\ket{\vb{x}}  e^{-i\vb{k}\cdot\vb{x}} \grad( \bra{\vb{x}}\ket{i}(t'))$ for the photon's frequency $\omega_k$ and its polarized direction $\vu{\varepsilon}_{\vb{k}\lambda}.$ In this paper we use $\ket{i}$ as the Whittaker wavepacket and $\ket{f}$ as the bound states, but wavepackets for any physical system can be used, either by using their analytic expressions as done here, or going through a fully numerical approach.

\bibliography{HA_bibliography.bib}

\pagebreak
\widetext
\begin{center}
\textbf{\large Supplemental Materials: Shaping Long-lived Electron Wavepackets for Customizable Optical Spectra}
\end{center}
\setcounter{equation}{0}
\setcounter{figure}{0}
\setcounter{table}{0}
\setcounter{page}{1}
\makeatletter
\renewcommand{\theequation}{S\arabic{equation}}
\renewcommand{\thefigure}{S\arabic{figure}}
\renewcommand{\bibnumfmt}[1]{[S#1]}
\renewcommand{\citenumfont}[1]{S#1}
The Supplemental Material is organized as follows. In section I 
we discuss the solutions of the Schr\"{o}dinger equation that yield the Whittaker wavepackets and show that they are physical wavepackets. In section \ref{sec:s:dynamics} we describe the dynamics of the Whittaker wavepackets and their underlying mathematical properties. 
In section \ref{sec:s:fgr} we discuss the Fermi Golden Rule (FGR) formalism and calculations. In section \ref{sec:s:num} we explain our numerical experiments.

\section{Derivation and Properties of Whittaker wavepackets} \label{sec:s:whit}
In this section we will motivate the origin of the Whittaker constructions. There are two steps in our construction: to solve the Schr\"{o}dinger equation for the extended eigenstates and then to construct the wavepackets from superpositions of these states.

We look for spherically symmetric extended states, i.e. $l=0$ and from separation of variables, the angular part is the spherical harmonic $Y_{00}(\theta,\phi) = (1/2)\pi^{-1/2}$. For the radial part $f_{El}$, the Schr\"{o}dinger equation takes the following form 
\eqb \label{eq:schrodinger}
 -\frac{\hbar^2}{2m}\pdv[2]{}{r}(rf_{El})+\left( V(r) + \frac{\hbar^2 l(l+1)}{2mr^2}\right)\cdot (rf_{El})=E\cdot (rf_{El}).
\eqe
Equation \eqref{eq:schrodinger} yields both bound eigenstate solutions (for $E<0$) and extended eigenstate solutions (for $E>0$). Given the mass of the electron $m_e$ and the electric constant $\epsilon_0$ we use the Coulomb potential $V(r) = -e^2/4\pi\epsilon_0 r$. It is useful to write $u(r)=rf_{Elm}$ and define the dimensionless parameter $x=r/(a_0/2)$, where $a_0$ is the Bohr radius. Likewise, let $\kappa=ka_0$ be the dimensionless parameter from momentum $k$. Substituting the parameters into equation \eqref{eq:schrodinger} we obtain the following differential equation for $u$
\eqb \label{eq:unitfree}
\left(\pdv[2]{x}+\frac{1}{x}+\kappa^2 \right)u = 0.
\eqe
The crux of our analysis is understanding the solutions to equation \eqref{eq:unitfree}. Luckily, we can reduce it to a known differential equation after some algebraic manipulations. Namely, let $W=u$, $z'=2ikx$, $k'=-i/2\kappa$ and $m'=1/2$. Then \eqref{eq:unitfree} is equivalent to the following 
$$
W''+\left(-\frac{1}{4} + \frac{k'}{z'} + \frac{\frac{1}{4}-m'^2}{z'^2}\right)W=0, 
$$
which is known as a version of the \textit{Whittaker differential equation} in the literature (\cite{abramowitz1972, zwillinger1997}).
A basis for the solutions is the following expression
\eqb \label{eq:s:solution}
u_\kappa(x)=
\frac{2i\kappa xe^{-i\kappa x} }{\Gamma\left(1-\frac{i}{2\kappa}\right)\Gamma \left(1+\frac{i}{2\kappa}\right)}\int_0^1 e^{2i\kappa x s} s^{\frac{i}{2\kappa}}(1-s)^{-\frac{i}{2\kappa}} \dd s,
\eqe
and its conjugate $\ovrl{u}_\kappa(x)$. We need to divide by $x$ to obtain $f_{Elm}$, which yields 
the equation in the main text for the Whittaker modes
\eqb \label{main:eq:solution_schroedinger} w_\kappa(x,0) = \frac{4i\kappa^2 e^{-i\kappa x} }{\pi\csch\left( \pi / 2\kappa \right)}\int_0^1 e^{2i\kappa x s} {\left(\frac{s}{1-s}\right)}^{\frac{i}{2\kappa}} \dd s, \eqe
where $w_k(x,0)\equiv f_{Elm}(x)$, as desired and $w_\kappa(x,t)=w_\kappa(x,0)\exp(-i\omega t \kappa^2)$ where the time evolution frequency is given by $\omega = 2e^2/a_0\hbar \approx 82 \ \si{\femto \second^{-1}}.$
From \eqref{main:eq:solution_schroedinger} we obtained the Whittaker wavepackets in \eqb \label{main:eq:whittaker_packet} \Psi_{E,\Delta E}(\vb{r},t)=\mc{N} \int_{-\infty}^{\infty}  e^{-( \kappa-\mu)^2/2\sigma^2}w_\kappa(x,t)\dd \kappa, \eqe where $\mc{N}$ is a normalization constant, $\mu$ and $\sigma$ are the mean and spread (standard deviation) of momentum. In SI units the energy $E$ is parameterized as  $E(\kappa) = (2e^2/4\pi\epsilon_0 a_0)\kappa^2$ via the dimensionless $\kappa$. We denote $E=E(\mu)$ and $\Delta E$ for the spread (standard deviation).
\qed

We claim that the wave function \eqref{main:eq:whittaker_packet} is square-integrable and therefore its probability density can have a physical meaning. We prove this via the following result. 

\begin{thm}
The Whittaker wavepacket $\Psi_{E,\Delta E}$, defined in \eqref{main:eq:whittaker_packet}, is square-integrable, hence its probability density has a physical meaning. 
\end{thm}

Let $\Psi \equiv \Psi_{E,\Delta E}$ for convenience. It suffices to show that $\int_{0}^{\infty}| \Psi|^2 x^2 \dd x$ is finite. We care only for the large $x$ behavior since the integral for small $x$ gives a finite contribution. Ignoring constant factors, for large $x$ we can use the approximation 
$$
x\Psi \sim \int_0^{3\sigma} \dd \kappa \exp\left(-\frac{( \kappa-\mu)^2}{2\sigma^2} \right)\exp(-i\kappa x)\exp(i(1/2\kappa)\ln(x)).
$$
Using the fact that $(e^{-i\kappa x})' = -ixe^{-i\kappa x}$ we can apply integration by parts to multiply the integral by a factor of $1/x$ and get additional contributions from constants. Ignoring the constants, we apply integration by parts again to obtain another factor of $1/x$. In conclusion, we obtain the following 
\eqb \label{eq:s:decay}
x|\Psi| < \frac{\text{const}}{x} + \frac{\text{const}\ln{x}}{x}+\frac{\text{const}(\ln{x})^2}{x^2}. 
\eqe
Since the functions in $x$ decay sufficiently quickly, from equation \eqref{eq:s:decay} we get that $\int_0^{\infty}| \Psi |^2 x^2 \dd x$ is bounded from above, as desired. \qed 

\section{Dynamics of the Whittaker wavepackets} \label{sec:s:dynamics}
In this section we introduce two important arguments: we describe the nature of the equations for the spatial spread $\Delta r$ and the diffraction lifetime $\Delta t$; we argue for a property that signifies the quasi-shape-invariant stability of the wavepackets \eqref{main:eq:whittaker_packet}.

First, we can develop simple analytic tools to get an understanding of the evolution in time of the Whittaker superpositions. Then we can conjecture a functional form for $\Delta r$ and $\Delta t$ that would yield formulas \eqb \label{main:eq:spread}
    \Delta r \approx \frac{2.471 \ a_0}{\sqrt{(\Delta E) / \si{\electronvolt}}}.
    \eqe
and 
\eqb \label{main:eq:lifetime}
    \Delta t \approx  \frac{0.136 \ \si{\electronvolt} \cdot \si{\femto \second}}{\sqrt{(\Delta E) E}}.
\eqe
 after fitting to our simulations.

\begin{thm} \label{suppl:thm:anzats}
The Whittaker wavepacket \eqref{main:eq:whittaker_packet} can be approximated as a Gaussian function in position space with mean $\mu_x$ and standard deviation $\sigma_x$ satisfying 
\eqb \label{eq:s:analytics}
\mu_x(t) = 2\mu \omega t  \ \ \text{ and } \ \ \sigma^2_x(t) = \frac{1}{\sigma^2} + 4\sigma^2 \omega^2 t^2.
\eqe
Moreover, a natural functional anzats for the spatial spread $\Delta r$ \eqref{main:eq:spread} and the diffraction lifetime $\Delta t$ \eqref{main:eq:lifetime} is given as follows 
\eqb \label{eq:s:anzats}
\Delta r = \frac{\mrm{const}}{\sqrt{\Delta E}} \ \ \text{ and } \ \ \Delta t = \frac{\mrm{const}}{\sqrt{(\Delta E)E}}.
\eqe
\end{thm}

\begin{proof}
After a large $x$ approximation, the integrand in \eqref{main:eq:whittaker_packet} looks as follows: 
$$
\exp\left(-\frac{( \kappa-\mu)^2}{2\sigma^2} \right) \frac{u_\kappa(x)}{x} \exp(-i\omega t \kappa^2).
$$
We can pull the $1/x$ term out of the integration since it does not affect the spread nor the lifetime. 
Suppose we make a large $x$ approximation and thus replace the exact Whittaker solution $u_\kappa(x)$ in $\Psi_{E,\Delta E}$ with its plane wave approximation $\exp(-i\kappa t)$ . Without loss of generality we are left to consider the following integral
\eqb \label{eq:s:simplification}
\Psi_{\text{approx.}} =
\int_{-\infty}^{\infty} \dd \kappa \exp\left(-\frac{( \kappa-\mu)^2}{2\sigma^2} \right) e^{-i\kappa x} e^{-i \omega t \kappa^2}.
\eqe
Now, having the approximation \eqref{eq:s:simplification} we can produce an analytic estimate for $\Delta x$ (and thus for $\Delta r$). Equation \eqref{eq:s:simplification} combines the exponential with the $\kappa^2$ dependence to obtain $\exp(-a\kappa^2  + b \mu  + \text{const})$, where for convenience we denote $a = \left(1/2\sigma^2+i\omega t\right)$ and
$b = \kappa / \sigma^2$.
Performing a Fourier transform from momentum space to position space and ignoring the constants in the exponential (as they only change the amplitude and not the mean and the spread of the Gaussian), we obtain the following Gaussian in $x$: 
$$
\Psi_{\text{approx.}} \sim \exp\left( \frac{(b-ix)^2}{4a} \right) \sim \exp\left(
-\frac{(x-2\mu \omega t )^2}{2\left( \frac{1}{\sigma^2} + 4\sigma^2 \omega^2 t^2\right)}
\right).
$$
In the last line we substituted for $a$, $b$ and $c$ and ignored the phase factors and the constants that do not affect the evolution of the probability density $x^2|\Psi|^2$. From the last equation we extract the $\mu_x$ and $\sigma_x$ of the Gaussian in position space to obtain the statement \eqref{eq:s:analytics}. To conclude the theorem, observe from \eqref{eq:s:simplification} that $\Delta x = \sigma_x(0) = 1/\sigma.$ Converting to $r$ from the unitless $x$, and a conversion to units of energy yields the anzats. Furthermore, a stationary phase argument yields $\Delta x \sim 2\mu \omega \Delta t$. Hence, $\Delta t \propto 1/\mu \sigma$, which concludes the proof. 
\end{proof}
\noindent
Theorem \ref{suppl:thm:anzats} is the theoretical foundation for obtaining equations \eqref{main:eq:spread} and \eqref{main:eq:lifetime}, which govern the dynamics of the Whittaker wavepackets. In section \ref{sec:s:num} we describe the procedure of fitting the constants in the anzats \eqref{eq:s:anzats}.

Now, we claim that the mathematical properties of the Whittaker wavepackets of angular momentum zero can be used to explain their quasi-shape-invariance. 
As we show below, the Whittaker wave functions \eqref{main:eq:solution_schroedinger} are purely imaginary. This property results in the Whittaker wavepacket at time zero \eqref{main:eq:whittaker_packet} $\Psi(r,0)$ also being purely imaginary. Thus, the nodes of the probability density $r^2\Psi(r,0)^2$ are the double roots of the zeros of the real wave function $\Im \Psi(r,0)$. Those zeros of the imaginary part of $\Psi(r,0)$ are closely related to the zeros of the Whittaker functions, \eqref{main:eq:solution_schroedinger} described in \cite{dunster1988, gabutti2001}. In \ref{suppl:fig:interference-explained} we show that the zeros of the extended states are closely spaced near the origin as we vary $\kappa$ slowly. Therefore, in between the regions of vanishing (circled on figure \ref{suppl:fig:interference-explained}), $\Im \Psi(r,0)$ would take alternating signs. Hence, by the Intermediate Value Theorem, for small distance $r$, the wavepacket is forced to have a node near the nodes of each of the extended states.

\begin{figure*}[htb]
    \includegraphics[scale=0.5]{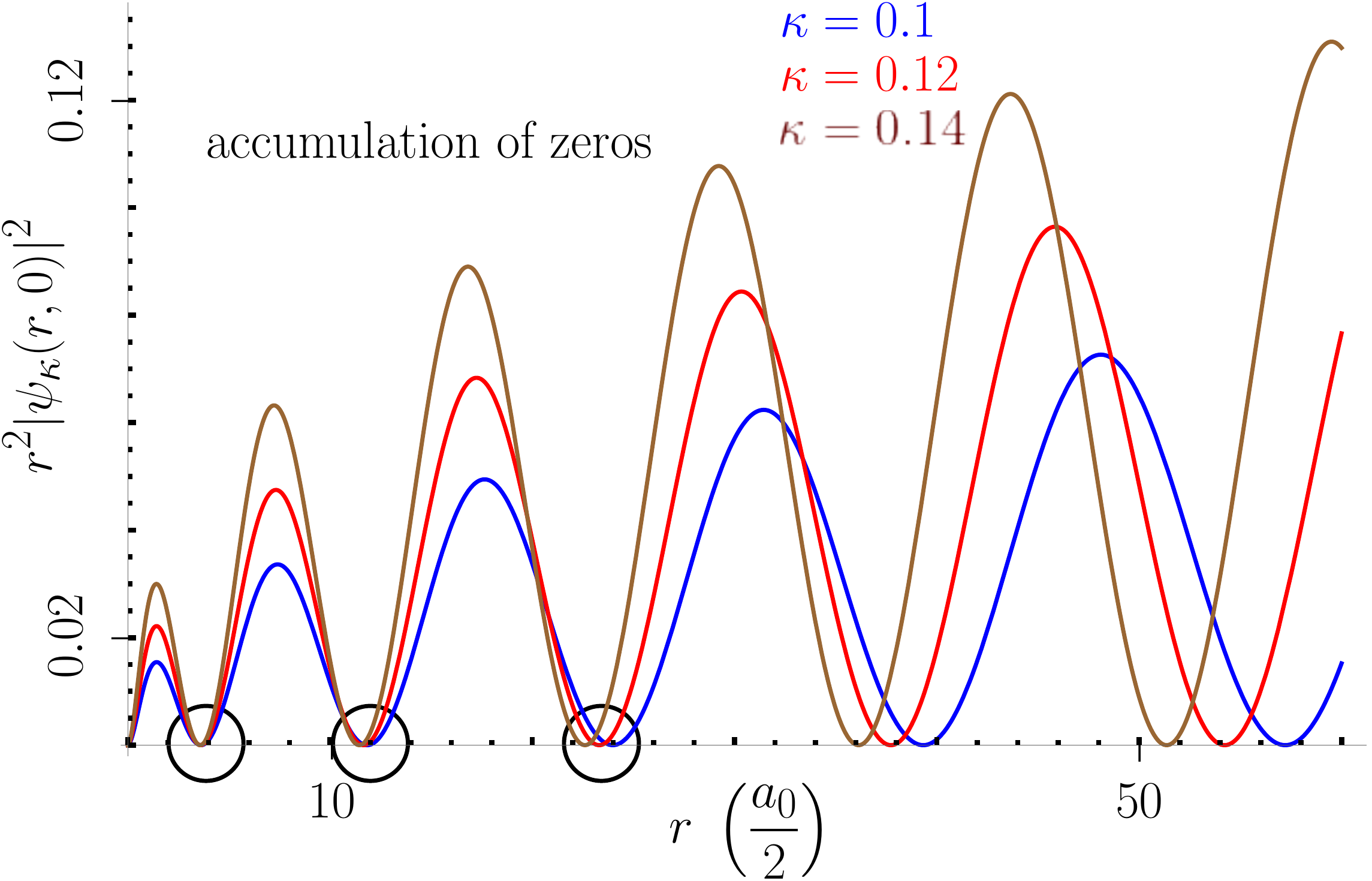}
    \caption{
        \textbf{The zeros of the Whittaker modes generate the nodes of the Whittaker wavepackets at time $t=0$.}  The zeros of these functions are close to each other for small $r$ and deviate from each other as $r$ becomes larger.  
        \label{suppl:fig:interference-explained}
    }
\end{figure*}

Therefore, it suffices to show that the functions \eqref{eq:s:solution} are purely imaginary. It is convenient to plug the exponential $e^{-i\kappa x}$ from the numerator in \eqref{eq:s:solution} into the integral and consider the following resulting expression
\eqb \label{eq:int1}
\int_0^1 e^{i\kappa x(2s-1)} s^{\frac{i}{2\kappa}}(1-s)^{-\frac{i}{2\kappa}} \dd s .
\eqe
The trick we present here is a change of variables of the following form
\eqb \label{eq:trick1}
s = 1-s'.
\eqe
Going through the algebra and relabeling $s'$ back to $s$, we obtain that expression \eqref{eq:int1} is equivalent to the following 
\eqb \label{eq:int2} 
\int_0^1 e^{i\kappa x(1-2s)} s^{-\frac{i}{2\kappa}}(1-s)^{\frac{i}{2\kappa}} \dd s .
\eqe
The symmetry of \eqref{eq:int1} is key to the proof that follows. 
\begin{thm}
The functions $w_\kappa(x,0)$ are purely imaginary.
\end{thm}
\begin{proof}
In equation \eqref{eq:s:solution} we plug the exponential from the numerator to factor out expression \eqref{eq:int1}. We are left with $2i\kappa x$ in the numerator, which is purely imaginary. The denominator is $\Gamma\left(1-\frac{i}{2\kappa}\right)\Gamma \left(1+\frac{i}{2\kappa}\right)$. By conjugating the Gamma function we see that this product is real. Hence, it suffices to show that expression \eqref{eq:int1} is real. We conjugate it and obtain equation \eqref{eq:int2}. The same equation came from the \emph{change-of-variables} trick \eqref{eq:trick1}, which means that the integral equals its conjugate, hence it is real, as desired. 
\end{proof}

\section{FGR Formalism For Whittaker wavepackets} \label{sec:s:fgr}
The goal of this section is to determine the transition rates from the Whittaker wavepackets \eqref{main:eq:whittaker_packet} to the bound states of the hydrogen atom. Our approach is based on computing the matrix elements of the $S$-matrix \cite{peskin1995}.
The $S$-matrix is given through the matrix elements of the time-ordered unitary evolution operator as follows
$$
S_{\text{fi}}=\bra{\mrm{f}} T \exp\left[ -\frac{i}{\hbar} \int_0^t \hat{H}_{\text{int}} \dd t'\right] \ket{\mrm{i}},
$$
for an initial state $\ket{i}$ and a final state $\ket{f}$, where the interaction Hamiltonian is given by
\eqb
\hat{H}_{\text{int}} = \int \dd^3 \vb{x} \psi^* \hat{H}_{\text{para}}(\vb{x},t) [\psi] = -\frac{i\hbar e}{m_e} \int \dd^3 \vb{x} \psi^* \vu{A}(\vb{x},t) \cdot \grad \psi.
\eqe
Then the infinitesimal probability of transition from $\ket{i}$ to $\ket{f}$ is given by the following equation 
\eqb \label{eq:s:prob}
\dd P_{\mrm{fi}}(\vb{k},\lambda) = \frac{V\dd^3 \vb{k}}{(2\pi)^3}|S_{\mrm{fi}}(\vb{k},\lambda)|^2,
\eqe
where $V$ is a finite volume needed for defining our measure. 
We would like to integrate (and sum) over all possible transition $\ket{\mrm{i}} \to \ket{\mrm{f}}$ involving the emission of a photon $\gamma(\vb{k},\lambda)$. 
In the Heisenberg picture, the vector potential looks as follows 
$$
\vu{A}(\vb{x},t) = \sum_{\vb{k},\lambda=1,2} \sqrt{\frac{\hbar}{2\epsilon_0 \omega_{\vb{k}} V}} \left( e^{i(\vb{k}\cdot \vb{x} - \omega_{k}t)} a_{\vb{k}\lambda}\vu{\varepsilon}_{\vb{k}\lambda} + 
e^{-i(\vb{k}\cdot \vb{x} - \omega_{k}t)} a^\dag_{\vb{k}\lambda}\vu{\varepsilon}^*_{\vb{k}\lambda}
\right),
$$
where there is a photon with momentum $\vb{k}$ and polarization $\lambda$, and the $a$-operator, with its conjugate $a^\dag$, are respectively the annihilation and creation operators for the same photon. The photon has a frequency $\omega_k$ and a polarized direction $\hat{\varepsilon}_{\vb{k}\lambda}.$ We are interested in the coupling of the atom to the EM field, so we concentrate on the paramagnetic term of the total Hamiltonian, which is given as follows 
$$
\hat{H}_{\text{para}}(\vb{x},t) = \frac{e}{m_e}\vu{A}(\vb{x},t) \cdot \vu{p},
$$
where $m_e$ is the mass of the electron. Suppose we look at the spontaneous emission from the initial state $\ket{\mrm{i}}$ to the final state $\ket{\mrm{f}}$ by emitting a photon.

Next we simplify the $S$-matrix. For $\bra{\vb{x}}\ket{f}(t) = \Psi_{\mrm{fin.}}(\vb{x},t)$ and $\bra{\vb{x}}\ket{i}(t)=\Psi_{\mrm{in.}}(\vb{x},t)$, up to first order, we obtain the following
$$
S_\mrm{fi}(\vb{k},\lambda)=-\frac{e}{m_e}\sqrt{\frac{\hbar}{2\epsilon_0 \omega_{k} V}} \int_0^t \dd t' e^{i\omega_{k}t'} \hat{\varepsilon}_{\vb{k}\lambda} \cdot
\int \dd^3 \vb{x} \Psi_{\mrm{fin.}}(\vb{x},t')^* \exp(-i\vb{k}\cdot\vb{x}) \grad \Psi_{\mrm{in.}}(\vb{x},t'). 
$$
The last equation yields the following expansion of $|S_\mrm{fi}(\vb{k},\lambda)|^2$ as an Einstein summation
\eqb \label{eq:ein_sum}
\frac{e^2\hbar}{2m_e \epsilon_0 \omega_{k} V} \int_0^t \dd t' \int_0^t \dd t'' e^{i\omega_k(t'-t'')} \cdot  
\int \dd^3 \vb{x} \dd^3 \vb{y} \varepsilon_{\vb{k}\lambda}^i \varepsilon_{\vb{k}\lambda}^j \Psi_{\mrm{fin.}}^*(\vb{x},t') \Psi_{\mrm{fin.}}(\vb{y},t') \exp(i\vb{k}\cdot (\vb{y} - \vb{x})) \cdot 
\pdv{\Psi_{\mrm{in.}}(\vb{x},t')}{x_i} \pdv{\Psi^*_{\mrm{in.}}(\vb{y},t')}{y_j}  .
\eqe
Now, we rely on the following chain of basic derivations
\begin{align}
\hat{k} = \frac{\vb{k}}{|\vb{k}|}; & \ \ \ \text{normalization}\\
\label{eq:s:ortho}
\hat{\varepsilon}_{\vb{k}\lambda_1} \otimes \hat{\varepsilon}_{\vb{k}\lambda_1} +  
\hat{\varepsilon}_{\vb{k}\lambda_2} \otimes \hat{\varepsilon}_{\vb{k}\lambda_2} + 
\hat{k} \otimes \hat{k} = \vb{1}_3; & \ \ \ \text{orthogonality}\\ 
\label{eq:s:ein}
\sum_{\lambda} \varepsilon_{\vb{k}\lambda}^i \varepsilon_{\vb{k}\lambda}^j + \hat{k}_i \hat{k}_j = \delta_{ij}; & \ \ \ \text{component-wise \eqref{eq:s:ortho}}\\
\label{eq:s:r}
\sum_{\lambda} \varepsilon_{\vb{k}\lambda}^i \varepsilon_{\vb{k}\lambda}^j = \delta_{ij} - \frac{\hat{k}_i \hat{k}_j}{|\vb{k}|^2}; & \ \ \ \text{rewriting \eqref{eq:s:ein}}\\ 
\label{eq:v:1}
v_i = \int \dd^3 \vb{x} \exp(-i \vb{k} \cdot \vb{x}) \Psi_{\mrm{fin.}}^*(\vb{x},t')\pdv{\Psi_{\mrm{in.}}(\vb{x},t)}{x_i}; & \ \ \ \text{extracted from equation \eqref{eq:ein_sum}}\\ 
\label{eq:v:2}
v^*_j = \int \dd^3 \vb{y} \exp(-i \vb{k} \cdot \vb{y}) \Psi_{\mrm{fin.}}(\vb{y},t')\pdv{\Psi_{\mrm{in.}}^*(\vb{y}, t')}{y_j}; & \ \ \ \text{extracted from equation \eqref{eq:ein_sum}}\\ 
v_i(\delta_{ij} - \hat{k}_i \cdot \hat{k}_j) v_j^* = |\vb{v}|^2-|\vb{v}\cdot \hat{k}|^2=|\vb{v}\cross \hat{k}|^2. \label{eq:fgr:chain:last}  & \ \ \ \text{combining \eqref{eq:s:r}, \eqref{eq:v:1} and \eqref{eq:v:2}}
\end{align}
Finally, to sum over the polarizations \eqref{eq:s:prob}, we need the last step from the chain of the derivations to obtain 
\eqb \label{eq:s:gen}
\sum_{\lambda} |S_\mrm{fi}(\vb{k},\lambda)|^2 = \frac{e^2\hbar}{2m_e^2\epsilon_0 \omega_{k}V}\bigg| \int_0^t \dd t' \int \dd^3 \vb{x}  e^{i\omega_{k}t'}\Psi_{\mrm{fin.}}^*(\vb{x}, t') \cdot \\ e^{-i\vb{k}\cdot \vb{x}}\left(\hat{k}\cross \grad \Psi_{\mrm{in.}}(\vb{x},t')\right) \bigg| ^2.
\eqe

Equation \eqref{eq:s:gen} is the most general formula in our analysis. From now on we assume that $\Psi_\mrm{in.}$ is the Whittaker wavepacket \eqref{main:eq:whittaker_packet} and that $\Psi_\mrm{fin.}$ is the standard bound state ${\psi}^{*}_{nlm}(\vb{x})e^{-i\omega_n t}$. Thus, in our case, $\ket{i}$ is given by \eqref{main:eq:whittaker_packet}, which means that the initial state is spherically symmetric, hence the quantum number $l$ is zero and thus has no angular dependence. Applying those remarks to \eqref{eq:s:gen} we directly get 
$$
\sum_{\lambda} |S_\mrm{fi}(\vb{k},\lambda)|^2 = \frac{e^2\hbar}{2m_e^2\epsilon_0 \omega_{k}V}\bigg| \int_0^t \dd t' \int \dd^3 \vb{x}  e^{i(\omega_{k}+\omega_n)t'}\psi^*_{nlm}(\vb{x}) \cdot \\ e^{-i\vb{k}\cdot \vb{x}}\hat{k}\cross \grad \Psi_{E,\Delta E}(\vb{x},t) \bigg| ^2.
$$

To further simplify this formula we express both the position $\vb{x}$ and momentum $\vb{k}$ vectors in spherical coordinates: $(r,\theta_{x},\phi_{x})$ and $(k,\theta_{k},\phi_{k}).$ We use the formula for wavepackets \eqref{main:eq:whittaker_packet} and the decomposition of the bound state into a radial part and a spherical harmonic: $\psi_{nlm}=R_{nl}Y_{lm}$, as parameterized in \cite{griffiths2004}. The azimuth dependence for $\vb{k}$ is trivially $2\pi$. Then after simplifying, and integrating \eqref{eq:s:prob} over $\vb{k}$ and using the conventions from the previous paragraph, the formula for the Whittaker's probability of decay becomes as follows
\begin{multline} \label{eq:s:fgr_num}
P(t) = \frac{e^2\hbar}{16m_e^2\epsilon_0 c \pi^3 a_0^4}  \int_0^{\infty} k \dd k  \int_0^{2\pi} \dd \phi_k \int_0^\pi \sin \theta_{k} \dd \theta_{k} \cdot 
\left| \int \dd \phi_x \sin \theta_{x} \dd \theta_{x} (\hat{k} \cross \hat{x})Y_{lm}^*(\theta_k,\phi_k)e^{-i\vb{k} \cdot \vb{x}}\right|^2 \\
\Bigg| \int_{-\infty}^{\infty} \dd z \left( e^{-(z-\mu)^2/2\sigma^2}
\cdot
\left[ \int_0^\infty r^2 \dd r R_{nl}^*(r) \pdv{w_z(r/a_0,0)}{r}\right] 
\frac{e^{(\omega_k+\omega_n-\omega z^2)t}-i}{\omega_k+\omega_n-\omega z^2}\right) \Bigg|^2,
\end{multline}
where $\vb{k} \cdot \vb{x} = kr(\sin \theta_x \sin \theta_k \cos (\phi_x - \phi_k) + \cos \theta_x \cos \theta_k)$ and in the standard basis $\hat{k} \cross \hat{x} = (\sin \phi_k \sin \theta_k \cos \theta_x - \sin \phi_x \sin \theta_x \cos \theta_k, \cos \phi_x \sin \theta_x \cos \theta_k - \sin \phi_k \sin \theta_k \cos \theta_x, \sin \theta_x \sin \theta_k \sin (\phi_x - \phi_k)).$

\section{Numerical experiments} \label{sec:s:num}

Our numerical experiments are synthesized in modules in the software \textit{Mathematica}, available at \cite{dangovski2017}. In this section we discuss methods for calculating the parameters $\Delta r$ \eqref{main:eq:spread}, $\Delta t$ \eqref{main:eq:lifetime} and the average rate $\tilde{\Gamma}=P\left(2(\Delta t)\right)/2(\Delta t)$.

\begin{itemize}
\item{Diffraction lifetime $\Delta t$:
to compute the lifetime we need to evaluate the overlap function 
and then fit into the anzats given by Theorem \ref{suppl:thm:anzats}. By definition we need to integrate for $r$ in the range $(0,\infty)$. In practice, this integration can be done efficiently by observing that the overlap $\mathbb{O}(t)$ is well-approximated by a Gaussian form. Moreover, by a numerical experiment we find that the shape of the decay does not change significantly if we only evaluate the integral up to some finite number. Therefore, we chose $(0,5a_0)$ as the range of integration. The procedure yields the constant of proportionality in \eqref{main:eq:lifetime}. The uncertainty coming from our fits is less than 5\% and hence it does not affect the claims of the paper.
}
\item{Spatial spread $\Delta r$:
to compute $\Delta r$ we observe that the upper envelope of the wavepackets at time zero is converging to the right half of a Gaussian as $\Delta E$ tends to zero. Hence, we extract the envelope numerically and then we fit a Gaussian for the data points. 
}
\item{Average rate of decay $\tilde{\Gamma}$:
to compute $\tau$ we first evaluate the integral \eqref{eq:s:fgr_num}. Our computations allow us to plot $P(t)$ as a function of time. In order to simplify our calculations we assume the dipole approximation. Note that this approximation is accurate for transition energies in the visible spectrum, but becomes less accurate for transition energies in the soft and hard X-rays. Moreover, the approximation is useful for better understanding the results because it enables a selection rule for the quantum number $l$ of the bound state $\psi_{nlm}$: $l$ = 1 since $\Psi_{E,\Delta E}$ is spherically symmetric and hence $m=-1,0,1$ (the probability of decay is independent of the selected $m$). Another interesting feature for the plot of $P(t)$ (shown in figure \ref{suppl:table:fgr}) is that there is a steady state for the probability, i.e. after a certain point in time, the rate of decay becomes zero. We can explain this phenomenon from a physical point of view: we know that the Whittaker wavepackets spreads with time. Hence, after some time the electron will be far away from the hydrogen atom range, so its overlap with the bound states vanishes. Figure \ref{suppl:table:fgr} also shows the decay probability for a range of energies and their energy spreads up to a transition energy at the soft X-ray spectrum. The plots start with a quadratic behavior and then switch to a linear regime before they reach their steady state. Note that this behavior of FGR probability resembles the one for transitions between bound states. 
}
\end{itemize}

\begin{figure}[ht]
  \begin{adjustbox}{addcode={\begin{minipage}{\width}}{\end{minipage}},rotate=90,center}
      \includegraphics[scale=0.65]{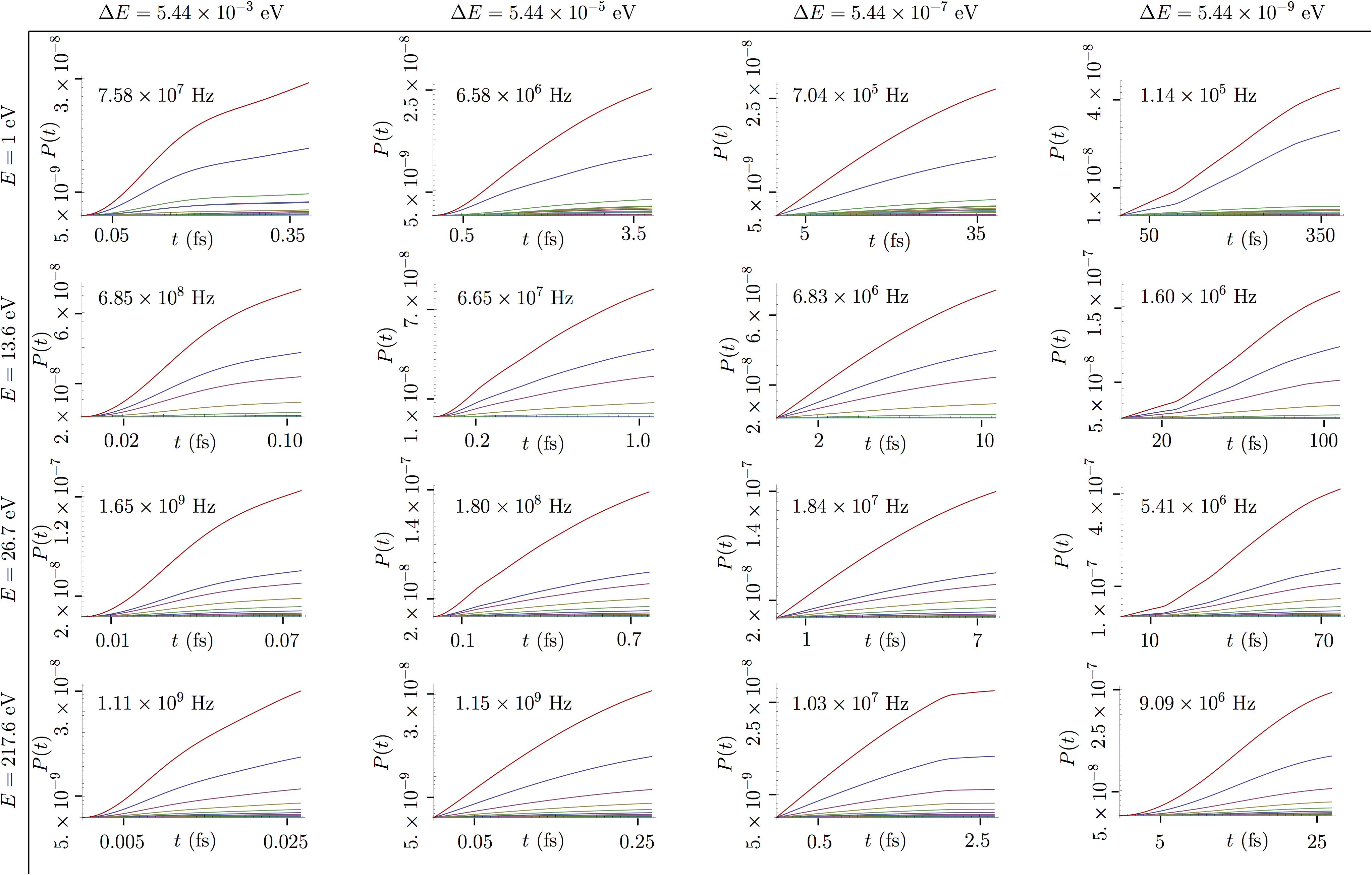}%
  \end{adjustbox}
  \caption{\textbf{The spontaneous emission dynamics of our wavepacket to bound states, marking the average transition rates.}
The transition dynamics is obtained based on the FGR formalism \eqref{eq:s:prob}, exhibiting a monotonous patterns in the case of varying $\Delta E$. The dipole approximation is used to simplify the calculation (the formalism can be applied more generally as discussed in section \ref{sec:s:fgr}).
\label{suppl:table:fgr}
}
\end{figure}

\end{document}